\documentclass[11pt]{article}

\usepackage{latexsym,dsfont,textcomp,amsmath}
\usepackage[english]{babel}
\usepackage[latin1]{inputenc}
\usepackage{enumerate,indentfirst}
\usepackage[colorlinks,citecolor=blue,urlcolor=blue]{hyperref}
\usepackage{amssymb}
\usepackage{amsthm}
\usepackage{hyperref}

\newtheorem{Theorem}{Theorem}[part]
\newtheorem{Definition}{Definition}[part]
\newtheorem{Proposition}{Proposition}[part]
\newtheorem{Assumption}{Assumption}[part]
\newtheorem{Lemma}{Lemma}[part]

\newtheorem{Remark}{Remark}[part]

\newcommand{\nc}{\newcommand}
\nc{\esssup}{\mathop{\mathrm{ess\,sup}}}
\nc{\essinf}{\mathop{\mathrm{ess\,inf}}}
\nc{\argmax}{\mathop{\mathrm{arg\,max}}}

\def \P{\mathbb{P}}
\def \N{\mathbb{N}}
\def \R{\mathbb{R}}

\def \E{\mathbb{E}}
\def \F{\mathbb{F}}
\def \G{\mathbb{G}}
\def \Q{\mathbb{Q}}

\def \1{\mathds{1}}

\def \Ac{{\cal A}}

\def \Fc{{\cal F}}
\def \Gc{{\cal G}}

\def \Sc{{\cal S}}

\def \Wb{\bar{W}}
\def \Mb{\bar{M}}

\def \Wt{\tilde{W}}

\begin{document}
\title{Portfolio Optimization in a Default Model under Full/Partial Information}
\author{Thomas LIM\\ \small{ENSIIE} \\ \small{Laboratoire d'Analyse et Probabilit\'es }\\ \small{Universit{\'e} d'Evry-Val-d'Essonne} \\ \small{e-mail: lim@ensiie..fr} \and Marie-Claire QUENEZ \\ \small{Laboratoire de Probabilit\'es et} \\ \small{Mod\`eles Al\'eatoires} \\\small{CNRS, UMR 7599}\\ \small{Universit{\'e} Paris 7}\\ \small{e-mail: quenez@math.jussieu.fr}}

\maketitle

\begin{abstract}
In this paper, we consider a financial market with an asset exposed to a risk inducing a jump in the asset price, and which can still be traded after the default time. We use a default-intensity modeling approach, and address in this incomplete market context the problem of maximization of expected utility from terminal wealth for logarithmic, power and exponential utility functions. We study this problem as a stochastic control problem both under full and partial information. Our contribution consists of showing that the optimal strategy can be obtained by a direct approach for the logarithmic utility function, and the value function for the power (resp. exponential) utility function can be determined as the minimal (resp. maximal) solution of a backward stochastic differential equation. For the partial information case, we show how the problem can be divided into two problems: a filtering problem and an optimization problem. We also study the indifference pricing approach to evaluate the price of a contingent claim in an incomplete market and the information price for an agent with insider information.
\end{abstract}


\newpage

\section{Introduction}
\setcounter{equation}{0} \setcounter{Assumption}{0}
\setcounter{Theorem}{0} \setcounter{Proposition}{0}
\setcounter{Corollary}{0} \setcounter{Lemma}{0}
\setcounter{Definition}{0} \setcounter{Remark}{0}

One of the important problems in mathematical finance is the portfolio optimization problem when the investor wants to maximize the expected utility from terminal wealth. In this paper, we study this problem for the classical utility functions by considering a small investor on an incomplete financial market who can trade in a finite time interval $[0,T]$ by investing in risky stocks and a riskless bond. We assume that there exists a default time on the market, and this one generates a jump of stock price. The underlying traded asset is assumed to be a local martingale driven by a Brownian motion and a default indicating process. We assume that in the market there are two kinds of agents: the insider agents (the agents with insider information) and the classical agents (they only observe the asset prices and the default times). These situations are referred as full information and partial information. We will be interested not only in describing the investor's optimal utility, but also the strategies which he may follow to reach this goal. \\
\indent The utility maximization problem with full information has been largely studied in the literature. In the framework of a continuous-time model the problem was studied for the first time by Merton \cite{mer71}. Using the methods of stochastic optimal control, the author derives a nonlinear partial equation for the value function of the optimization problem. Some papers study this problem by using the dual problem, we can quote, for instance, Karatzas \emph{et al.} \cite{karlehshr87} for the case of complete financial models, and Karatzas \emph{et al.} \cite{karlehshrxue91} and Kramkov and Schachermayer \cite{krasch99} for the case of incomplete financial models, they find the solution of the original problem by convex duality. These papers are useful to prove the existence of an optimal strategy in the general case, but in practice it is difficult to find the optimal strategy with the dual method. Some others study the problem by using the dynamic programming principle, we can quote Jeanblanc and Pontier \cite{jeapon90} for a complete model with discontinuous prices, Bellamy \cite{bel01} in the case of a filtration generated by a Brownian motion and a Poisson measure, Hu \emph{et al.} \cite{huimkmul05} for an incomplete model in the case of a Brownian filtration,  and Jiao and Pham \cite{jiapha09} in the case with a default, in which the authors study the case before the default and the case after the default.\\
\indent Models with partial observation are essentially studied in the literature in a complete market framework. Detemple \cite{det86}, Dothan and Feldman \cite{dotfel86}, Gennotte \cite{gen86} use dynamic programming methods in a linear gaussian filtering. Lakner \cite{lak95,lak98} solves the optimization problem via a martingale approach and works out the special case of linear gaussian model. We mention that Frey and Runggaldier \cite{frerun99} and Lasry and Lions \cite{laslio99} study hedging problems in finance under restricted information. Pham and Quenez \cite{phaque01} treat the case of an incomplete stochastic volatility model. Callegaro \emph{et al.} \cite{calmasrun06} and Roland \cite{rol07} study the case of a market model with jumps.\\
\indent We first study the case of full information. For the logarithmic utility function, we use a direct approach, which allows to give an expression of the optimal strategy depending uniquely on the coefficients of the model satisfied by the stocks. For the power utility function, we look for a necessary condition characterizing the value function which is solution of the maximization problem. We show that this value function is the smallest solution of a backward stochastic differential equation (in short BSDE). We also give an approximation of the value function by a sequence of solutions of BSDEs. These solutions are the value functions of the maximization problem restricted to some bounded subsets of strategies. For the exponential utility function, we refer to the companion paper Lim and Quenez \cite{limque09}.\\
\indent In order to solve the partial information problem, the common way is to use the filtering theory, so as to reduce the stochastic control problem with partial information to one with full information as in Pham and Quenez \cite{phaque01} or Roland \cite{rol07}. Then we can apply the results of the full information problem.\\
\indent The outline of this paper is organized as follows. In Section 2, we describe the model and formulate the optimization problem. In Section 3, we solve the maximization problem for the logarithmic utility function with a direct approach.  In Section 4, we consider the power utility function by giving a characterization of the value function by a BSDE thanks to the dynamic programming principle, then we approximate the value function by a sequence of solutions of Lipschitz BSDEs. In Section 5, we use results from filtering theory to reduce the stochastic control problem with partial information to one with full information, then we apply the results of the full information problem to the partial information problem. Finally we study the indifference price for a contingent claim and the information price linked to the insider information.\\

\section{The model}\label{modele}
\setcounter{equation}{0} \setcounter{Assumption}{0}
\setcounter{Theorem}{0} \setcounter{Proposition}{0}
\setcounter{Corollary}{0} \setcounter{Lemma}{0}
\setcounter{Definition}{0} \setcounter{Remark}{0}

We start with a complete probability space $(\Omega, \mathcal{F},\mathbb{P})$ and $\mathbb{F}=\{\mathcal{F}_t\}_{0\leq t \leq T}$ a filtration in $\Fc$ satisfying the usual conditions (augmented and right-continuous). The terminal time $T < \infty$ is a fixed constant, and we assume throughout that all processes are defined on the finite time interval $[0,T]$. Suppose that this space is equipped with two stochastic processes: a Brownian motion $W$ and a jump process $N$ equal to $N_t=\mathds{1}_{\tau \leq t}$, where $\tau$ is a default time. We assume that this default time can appear at any time: $\P(\tau > t) > 0$. We denote by $M$ the compensated martingale of this process $N$ and by $\Lambda$ its compensator in the filtration $\F$. We assume that the compensator $\Lambda$ is absolutely continuous with respect to the Lebesgue measure, so that there exists a process $\lambda$ such that $\Lambda_t=\int_0^t\lambda_s ds$. We can see that 
\begin{equation}\label{M}
M_t=N_t-\int_0^t\lambda_sds,
\end{equation}
is an $\F$-martingale. In the sequel we assume that the process $\lambda$ is uniformly bounded. It should be noted that the construction of such process $N$ is fairly standard; see, for example, Bielecki and Rutkowski \cite{bierut04}. \\

\noindent We introduce some sets used throughout the paper:
\begin{itemize}
\item $L^{1,+}$ is the set of positive $\F$-adapted c\`ad-l\`ag processes on $[0,T]$ such that $\E[Y_t]<\infty$ for any $t\in [0,T]$.
\item $\Sc^2$  is the set of positive $\F$-adapted c\`ad-l\`ag processes on $[0,T]$ such that $\E[\sup_{t\in[0,T]} |Y_t|^2]<\infty$.
\item $L^2(W)$ (resp. $L^2_{loc}(W)$) is the set of $\F$-predictable processes on $[0,T]$ such that
\begin{equation*}
\E\Big[\int_0^T |Z_t|^2dt\Big]<\infty~~\text{(resp. }\int_0^T|Z_t|^2dt<\infty ~a.s.\text{ ).}
\end{equation*}
\item $L^2(M)$ (resp. $L^1_{loc}(M)$) is the set of $\F$-predictable processes on $[0,T]$ such that 
\begin{equation*}
\E\Big[\int_0^T\lambda_t |U_t|^2 dt\Big]<\infty~~\text{(resp. } \int_0^T\lambda_t |U_t| dt<\infty ~ a.s.\text{ ).}
\end{equation*}
\end{itemize}

We consider a financial market consisting of one risk-free asset, whose price process is assumed for simplicity to be equal to $1$ at each date, and one risky asset with a price process $S$ evolving according to the following diffusion
\begin{equation}
dS_t= S_{t^-} (\mu_t dt + \sigma_t dW_t + \beta_{t} dN_t),~~0 \leq t \leq T \,.
\end{equation}
We shall make the following standing assumptions:
\begin{Assumption}\label{hypothese coeff}
{\rm \begin{itemize}
\item $\mu$ and $\sigma$ are $\R$-valued uniformly bounded adapted stochastic processes, with  $\sigma_t > 0$ and $\theta_t = \mu_t / \sigma_t$ uniformly bounded.
\item $\beta$ is a $\R$-valued uniformly bounded predictable stochastic process, with $\beta_t > -1$ for any $t \in [0,\tau]$.
\end{itemize}
}\end{Assumption}
\noindent The last assumption implies that the process $S$ is positive.\\

An $\F$-predictable process $\pi=(\pi_t)_{0\leq t \leq T}$ is called trading strategy  if $\int\frac{\pi_t X_t}{S_t} dS_t$ is well defined where $X_t$ is the wealth at time $t$. The process $\pi$ describes the part of the wealth invested in the risky asset. The number of shares of the risky asset is given by $\frac{\pi_t X_t}{S_t}$. The wealth process $X^{x,\pi}$ of a self-financing trading strategy $\pi$ with initial capital $x$ satisfies the equation
\begin{equation}\label{dolean n}
X_t^{x,\pi}=x \exp \Big( \int_0^t \big( \pi_s \mu_s - \frac{|\pi_s \sigma_s|^2}{2} \big) ds + \int_0^t \pi_s \sigma_s dW_s \Big) (1+\pi_{\tau} \beta_{\tau} N_t) \,.
\end{equation}
For a given initial time $t$ and an initial capital $x$, the associated wealth process is denoted by $X^{t,x,\pi}$.\\

\indent Now let $U:\R\rightarrow \R$ be a utility function. The optimization problem consists of maximizing the expected utility from terminal wealth over the class $\Ac(x)$ of admissible portfolios (which will be defined in the sequel). More precisely, we want to characterize the value function of this problem, which is defined by
\begin{equation}  \label{V}
V(x)=\sup\limits_{\pi \in \Ac(x)} \E\big[U(X_T^{x,\pi})\big] \,,
\end{equation}
and we also want to give the optimal strategy when this one exists. We begin by the simple case when $U$ is the logarithmic utility function, then we study the power and exponential utility functions.

\section{Logarithmic utility function}\label{section logarithme}
\setcounter{equation}{0} \setcounter{Assumption}{0}
\setcounter{Theorem}{0} \setcounter{Proposition}{0}
\setcounter{Corollary}{0} \setcounter{Lemma}{0}
\setcounter{Definition}{0} \setcounter{Remark}{0}
In this section, we specify the meaning of optimality for trading strategies by stipulating that the agent wants to maximize his expected utility from his terminal wealth $X_T^{x,\pi}$ with respect to the logarithmic utility function
\begin{equation*}
U(x)=\log(x),~x>0 \,.
\end{equation*}
Our goal is to solve the following optimization problem
\begin{equation}\label{optimisation log}
V(x)=\sup_{\pi\in\Ac(x)}\E\big[\log(X^{x,\pi}_T)\big] \,,
\end{equation}
with $\Ac(x)$ the set of admissible portfolios defined by:
\begin{Definition} \label{admissible logarithmique} \emph{
The set of admissible trading strategies $\Ac(x)$ consists of all $\F$-predictable processes $\pi$ satisfying $\E [\int_0^T|\pi_t\sigma_t|^2dt ] + \E [\int_0^T\lambda_t|\log(1+\pi_t\beta_t)|dt ] < \infty$, and such that $\pi_t\beta_t> -1$ a.s. for any  $0\leq t \leq \tau$.
}\end{Definition}
We can see from (\ref{optimisation log}) that $V(x)=\log(x)+ V(1)$. Hence, we only study the case $x=1$. And for the sake of brevity, we shall denote $X^\pi_t$ instead of $X^{1,\pi}_t$ and $\Ac$ instead of $\Ac(1)$.
By definition of $X^\pi$ we obtain
\begin{equation}\label{dolean}
\log(X^{\pi}_t)=\int_0^t\big(\pi_s\mu_s-\frac{|\pi_s\sigma_s|^2}{2}\big)ds+\int_0^t\pi_s\sigma_sdW_s+\int_0^t\log(1+\pi_s\beta_{s})(dM_s+\lambda_sds) \,.
\end{equation}
\noindent As in \cite{krasch99}, we assume that $\sup_{\pi\in \Ac}\E[\log(X_T^{\pi})]<\infty$.\\

We add the following assumption on the coefficients to be able to solve the optimization problem (\ref{optimisation log}) directly:
\begin{Assumption}\emph{\label{coeff inverse borne}
The process $\beta^{-1}$ is uniformly bounded.
}\end{Assumption}
With this assumption, we get easily the value function $V(x)$ and the optimal strategy:
 \begin{Theorem}
 The solution of the optimization problem (\ref{optimisation log}) is given by 
 \begin{equation*}
 V(x)=\log(x)+\E\Big [\int_0^T\big(\hat{\pi}_t\mu_t-\frac{|\hat{\pi}_t\sigma_t|^2}{2}+\lambda_t\log(1+\hat{\pi}_t\beta_t)\big)dt\Big],
  \end{equation*}
 with $\hat{\pi}$ the optimal trading strategy defined by
  \begin{equation}\label{strategie optimal}
 \hat{\pi}_t=\left\{\begin{aligned}
& \frac{\mu_t}{2\sigma_t^2}-\frac{1}{2\beta_t}+\frac{\sqrt{(\mu_t\beta_t+\sigma_t^2)^2+4\lambda_t\beta_t^2\sigma_t^2}}{2\beta_t\sigma_t^2} ~~\text{ if} ~t\leq \tau~\text{and } \beta_t\neq 0,  \\ 
& \frac{\mu_t}{\sigma_t^2} ~~\text{ if} ~t\leq  \tau~\text{and } \beta_t= 0 ~\text{or } t > \tau.
\end{aligned}\right.
 \end{equation}
 \end{Theorem}

\begin{proof}
With (\ref{dolean}) and Definition \ref{admissible logarithmique}, we get the following expression for $V(1)$
\begin{equation*}
V(1)=\sup_{\pi\in\Ac}\E\Big[\int_0^T\big(\pi_t\mu_t-\frac{|\pi_t\sigma_t|^2}{2}+\lambda_t\log(1+\pi_t\beta_t)\big)dt\Big],
\end{equation*}
which implies that 
\begin{equation}\label{inegalite V}
V(1)\leq\E\Big[\int_0^T\esssup_{\pi_t\beta_t>-1}\Big\{\pi_t\mu_t-\frac{|\pi_t\sigma_t|^2}{2}+\lambda_t\log(1+\pi_t\beta_t)\Big\}dt\Big].
\end{equation}
For any $t\in[0,T]$ and any $\omega\in \Omega$, we have
\begin{equation*}
\esssup_{\pi_t\beta_t>-1}\Big\{\pi_t\mu_t-\frac{|\pi_t\sigma_t|^2}{2}+\lambda_t\log(1+\pi_t\beta_t)dt\Big\}=\hat{\pi}_t\mu_t-\frac{|\hat{\pi}_t\sigma_t|^2}{2}+\lambda_t\log(1+\hat{\pi}_t\beta_t),
\end{equation*}
with $\hat{\pi}$ defined by (\ref {strategie optimal}). Then from inequality (\ref{inegalite V}), we can see that
\begin{equation*}
V(1)\leq \E\Big [\int_0^T\big(\hat{\pi}_t\mu_t-\frac{|\hat{\pi}_t\sigma_t|^2}{2}+\lambda_t\log(1+\hat{\pi}_t\beta_t)\big)dt\Big].
\end{equation*}
It now is sufficient to show that the strategy $\hat{\pi}$ is admissible. It is clearly the case with Assumption \ref{coeff inverse borne}. Thus the previous inequality is an equality
\begin{equation*}
V(1)= \E\Big [\int_0^T\big(\hat{\pi}_t\mu_t-\frac{|\hat{\pi}_t\sigma_t|^2}{2}+\lambda_t\log(1+\hat{\pi}_t\beta_t)\big)dt\Big],
\end{equation*}
and the strategy $\hat{\pi}$ is optimal.
\end{proof}

\begin{Remark}
\emph{Assumption \ref{coeff inverse borne} can be reduced to 
\begin{equation*}
\E\Big[\int_0^T|\hat{\pi}_t\sigma_t|^2dt\big]+\E\big[\int_0^T\lambda_t|\log(1+\hat{\pi}_t\beta_t)|dt\Big] < \infty \,.
\end{equation*}}
\end{Remark}

\begin{Remark}
\emph{Recall that in the case without default, the optimal strategy is given by $\tilde \pi_t=\mu_t / \sigma^2_t$. Thus, in the case of default, the optimal strategy can be written under the form
\begin{equation*}
\hat{\pi}_t=\tilde \pi_t-\varepsilon _t,
\end{equation*}
where $\varepsilon _t$ is an additional term given by
\[ \varepsilon _t=\left\{\begin{aligned}
& \frac{\mu_t}{2\sigma_{t}^2}+\frac{1}{2\beta_t}-\frac{\sqrt{(\mu_t\beta_t+\sigma_{t}^2)^2+4\lambda_t\beta_t^2\sigma_{t}^2}}{2\beta_{t}\sigma_{t}^2} ~~\text{ if} ~t \leq \tau~\text{and } \beta_t\neq 0 \,,  \\ 
&0 ~~\text{ if} ~t \leq \tau~\text{and } \beta_t= 0 ~\text{or } t > \tau \,.
\end{aligned}\right.\]
Note that if we assume that $\beta$ is negative (resp. $\beta$ is positive), i.e. the asset price $S$ has a negative jump (resp. a positive jump) at default time $\tau$, $\varepsilon $ is positive (resp. negative), i.e. the agent has to invest less (resp. more) in the risky asset than in the case of a market without default.  }
\end{Remark}

\section{Power utility}\label{section puissance}
\setcounter{equation}{0} \setcounter{Assumption}{0}
\setcounter{Theorem}{0} \setcounter{Proposition}{0}
\setcounter{Corollary}{0} \setcounter{Lemma}{0}
\setcounter{Definition}{0} \setcounter{Remark}{0}

In this section, we keep the notation of Section \ref{section logarithme} and we shall study the case of the power utility function defined by
\begin{equation*}
U(x)=x^\gamma \,,~~x\geq 0 \,,~~\gamma\in (0,1) \,.
\end{equation*} 
In order to formulate the optimization problem we first define the set of admissible trading strategies.
\begin{Definition}\emph{
The set of admissible trading strategies $\Ac(x)$ consists of all $\F$-predictable processes $\pi$ such that $\int_0^T |\pi_t \sigma_{t}|^2 dt + \int_0^T  \lambda_t |\pi_t\beta_t| dt < \infty$ a.s. and such that $\pi_{\tau} \beta_{\tau} \geq -1$ a.s.
}\end{Definition}

\begin{Remark}
\emph{From expression (\ref{dolean n}), it is obvious that the condition $\pi_{\tau} \beta_{\tau} \geq -1$ a.s. is equivalent to $X^{x,\pi}_t \geq 0$ a.s. for any $0 \leq t \leq T$.}
\end{Remark}

The portfolio optimization problem consists of determining a predictable portfolio $\pi$ which attains the optimal value
\begin{equation}\label{optimisation puissance}
V(x) = \sup_{\pi\in\Ac(x)} \E \big[ (X^{x,\pi}_T)^\gamma \big] \,.
\end{equation}

\noindent Problem (\ref{optimisation puissance}) can be clearly written as $V(x)=x^\gamma V(1)$.
Therefore, it is sufficient to study the case $x=1$. As in \cite{krasch99}, we assume that $\sup_{\pi\in \Ac(1)}\E[(X_T^{1,\pi})^\gamma]<\infty$. To solve the optimization problem, we give a dynamic extension of the initial problem. For any initial time $t\in[0,T]$, we define the value function $J(t)$ by the following random variable
\begin{equation*}
J(t)=\esssup\limits_{\pi\in \Ac_t(1)}\E\Big[\big(X_T^{t,1,\pi}\big)^\gamma\Big|\Fc_t\Big],
\end{equation*}
with $\Ac_t(1)$ the set of $\F$-predictable processes $\pi=(\pi_s)_{t\leq s \leq T}$ such that $\int_t^T |\pi_s \sigma_{s}|^2 ds + \int_t^T |\pi_s \beta_s|\lambda_s ds < \infty$ a.s. and such that $\pi_{\tau} \beta_{\tau} \geq -1$ a.s.

For the sake of brevity, we shall denote $X_s^\pi$ (resp. $X_s^{t,\pi}$) instead of $X_s^{1,\pi}$ (resp. $X_s^{t,1,\pi}$) and $\Ac$ (resp. $\Ac_t$) instead of $\Ac(1)$ (resp. $\Ac_t(1)$).

In the sequel, we will use the martingale representation theorem to characterize the value function $J(t)$:
\begin{Lemma}\label{theoreme representation}
Any $(\P,\F)$-local martingale has the representation 
\begin{equation}\label{equation representation}
m_t=m_0+\int_0^t a_s dW_s + \int_0^t b_sdM_s , ~\forall\,t\in[0,T] ~ a.s.\,,
\end{equation}
where $a\in L^2_{loc}(W)$ and $b\in L^1_{loc}(M)$. If $m$ is a square integrable martingale, each term on the right-hand side of the representation (\ref{equation representation}) is square integrable. 
\end{Lemma}

\subsection{Optimization over bounded strategies}
Before studying the value function $J(t)$, we study the value functions $(J^k(t))_{k\in \N}$ defined by
\begin{equation}
J^k(t)=\esssup \limits_{\pi \in \Ac^k_t} \mathbb{E}\Big[(X_T^{t,\pi})^\gamma\Big| \mathcal{F}_t\Big] \,,~\forall\,t\in[0,T]~a.s.\,,
\label{Jk}
\end{equation}
where $\mathcal{A}^k_t$ is the set of strategies of $\Ac_t$ uniformly bounded by $k$. That means that the part of the wealth invested in the risky asset has to be bounded by a constant $k$, which makes sense in finance, because the ratio of the amount of money invested or borrowed to the wealth must be bounded according to the financial legislation. 

Let us fix $k \in \N$. We want to characterize the value function $J^k(t)$ by a BSDE.
For that we introduce for any $\pi\in\Ac^k$ the c\`{a}d-l\`{a}g process $J^\pi$ defined for any $t \in [0,T]$ by
\begin{equation*}
J^\pi_t=\E\big[(X^{t,\pi}_T)^\gamma\big|\Fc_t\big] \,.
\end{equation*}
The family $\{J^\pi, \pi\in\Ac^k\}$ is uniformly bounded:
\begin{Lemma}\label {Jk borne} 
The process $J^\pi$ is uniformly bounded by a constant independent of $\pi$.
\end{Lemma}

\begin{proof}
Fix $t\in[0,T]$. We have
\begin{equation*}
J_t^{\pi} =\E\Big[\exp\Big(\gamma\int_t^T(\mu_s\pi_s-\frac{|\sigma_{s}\pi_s|^2}{2})ds+\int_t^T\gamma\sigma_{s}\pi_sdW_s \Big)(1+\pi_{\tau}\beta_\tau\mathds{1}_{t<\tau\leq T})^\gamma\Big|\Fc_t\Big] \,,
\end{equation*}
since the coefficients $\mu$, $\sigma$ and $\beta$ are supposed to be uniformly bounded, we can see that
\begin{equation*}
J_t^{\pi}\leq(1+k\,|\beta|_\infty)^\gamma\,\exp\Big((\gamma\,k\,|\mu|_\infty+\gamma^2\frac{(k\,|\sigma|_\infty)^2}{2})\,T\Big) \,.
\end{equation*}
\end{proof}

Classically, for any $\pi\in \Ac^k$, the process $J^{\pi}$ can be shown to be the solution of a linear BSDE. More precisely, there exist $Z^{\pi}\in L^2(W)$ and $U^{\pi}\in L^2(M)$, such that $(J^{\pi}, Z^{\pi}, U^{\pi})$ is the solution in $\Sc^2\times L^2(W)\times L^2(M)$ of the linear BSDE with bounded coefficients
\begin{equation}
\left\{\begin{aligned}
 -\,dJ^{\pi}_t~=&
\Big\{\gamma \pi_t(\mu_tJ^{\pi}_t+\sigma_tZ^{\pi}_t) +\frac{\gamma(\gamma-1)}{2}\pi_t^2\sigma_t^2J^{\pi}_t+\lambda_t((1+\pi_t\beta_t)^\gamma-1)(J^{\pi}_t+U^{\pi}_t)\Big\}dt \\& - Z^{\pi}_tdW_t-U^{\pi}_tdM_t,\\
J^{\pi}_T~=& \;1.
\end{aligned}\right.
\label{edsr jkpi}
\end{equation}

Using the fact that $J^k(t)=\esssup_{\pi\in\Ac^k_t} J_t^\pi$ for any $t\in[0,T]$, we derive that $J^k(.)$ corresponds to the solution of a BSDE, whose generator is the essential supremum over $\pi$ of the generators of $\{J^\pi,\pi\in\Ac^k\}$. More precisely,
\begin{Proposition}
The following properties hold:
\begin{itemize}
\item Let $(Y, Z, U)$ be the solution in $\Sc^2\times L^2(W)\times L^2(M)$ of the following Lipschitz BSDE 
\begin{equation}
\left\{\begin{aligned}
 -\,dY_t=& \esssup\limits_{\pi\in\mathcal{A}^k}\Big\{\gamma \pi_t(\mu_tY_t+\sigma_tZ_t)+\frac{\gamma(\gamma-1)}{2}\pi_t^2\sigma_t^2Y_t 
+\lambda_t((1+\pi_t\beta_t)^\gamma-1)(Y_t+U_t)\Big\}dt \\ &-Z_tdW_t-U_tdM_t
,\\
Y_T=&\;1.
\label{EDSR lipschitz}
\end{aligned}\right.
\end{equation}
Then, $J^k(t)=Y_t$ a.s. for any $t\in[0,T]$.
\item There exists a unique optimal strategy for $J^k(0)=\sup_{\pi\in\Ac^k}\E[(X_T^{\pi})^\gamma]$.
\item A strategy $\hat{\pi}\in \Ac^k$ is optimal for $J^k(0)$ if and only if it attains the essential supremum of the generator in (\ref{EDSR lipschitz}) $dt\otimes d\P-a.e.$
\end{itemize}
\label{unicite}
\end{Proposition}

\begin{proof}
Since for any $\pi \in \Ac^k$ there exist $Z^\pi\in L^2(W)$ and $U^\pi\in L^2(M)$ such that $(J^\pi, Z^\pi, U^\pi)$ is the solution of the BSDE
\begin{equation*}
-\,dJ^\pi_t = f^\pi(t,J_t^\pi,Z_t^\pi,U_t^\pi)dt - Z_t^\pi dW_t - U_t^\pi dM_t  ~;~J_T^\pi=1 \,,
\end{equation*}
with $f^\pi(s,y,z,u)= \frac{\gamma(\gamma-1)}{2}\pi_s^2\sigma_s^2y+\gamma\pi_s(\mu_sy+\sigma_sz)+\lambda_s\big((1+\pi_s\beta_s)^\gamma-1\big)(y+u)$. Let us introduce the generator $f$ which satisfies $ds\otimes d\P-a.e.$
\begin{equation*}
f(s,y,z,u)=\esssup\limits_{\pi\in \Ac^k}f^\pi(s,y,z,u) \,.
\end{equation*}
Note that $f$ is Lipschitz, since the supremum of affine functions, whose coefficients are bounded by a positive constant $c$, is Lipschitz with Lipschitz constant $c$. Hence, the BSDE with Lipschitz generator $f$
\begin{equation*}
-\,dY_t=f(y,Y_t,Z_t,U_t)dt - Z_tdW_t - U_tdM_t~;~Y_T=1,
\end{equation*}
admits a unique solution denoted by $(Y, Z, U)$.\\
By the comparison theorem in case of jumps $Y_t \geq J_t^\pi$, $\forall\, t\in[0,T]$ a.s. As this inequality is satisfied for any $\pi \in \Ac^k$, it is obvious that $Y_t \geq \esssup_{\pi\in\Ac^k} J_t^{\pi}$ a.s. Also, by applying a predictable selection theorem, one can easily show that there exists $\hat{\pi}\in \Ac^k$ such that for any $t\in[0,T]$, we have
\begin{multline*}
 \esssup\limits_{\pi\in\mathcal{A}^k}\Big\{\gamma \pi_t(\mu_tY_t+\sigma_tZ_t)+\frac{\gamma(\gamma-1)}{2}\pi_t^2\sigma_t^2Y_t+\lambda_t((1+\pi_t\beta_t)^\gamma-1)(Y_t+U_t)\Big\}\\
 = \gamma \hat{\pi}_t(\mu_tY_t+\sigma_tZ_t)+\frac{\gamma(\gamma-1)}{2}\hat{\pi}_t^2\sigma_t^2Y_t+\lambda_t((1+\hat{\pi}_t\beta_t)^\gamma-1)(Y_t+U_t).
 \end{multline*}
Thus $(Y, Z, U)$ is a solution of the BSDE (\ref{edsr jkpi}) associated with $\hat{\pi}$. Therefore by uniqueness of the solution of the BSDE (\ref{edsr jkpi}), we have $Y_t=J_t^{\hat{\pi}}$ and thus $Y_t=\esssup _{\pi \in \mathcal{A}_t^k} J^{\pi}_t=J_t^{\hat{\pi}}$, $\forall \, t\in[0,T]$ a.s.\\
The uniqueness of the optimal strategy is due to the strict concavity of the function $x \mapsto x^\gamma$.
\end{proof}

\subsection{General case}
In this part, we characterize the value function $J(t)$ by a BSDE, but the general case is more complicated than the case with bounded strategies and it needs more technical tools.  Note that the random variable $J(t)$ is defined uniquely only up to $\P$-almost sure equivalent and that the process $J(.)$ is adapted but not necessarily progressive. Using dynamic control technics, we derive the following characterization of the value function:

\begin{Proposition}\label{J(t) plus petite surmartingale}
$J(.)$ is the smallest $\F$-adapted process such that $(X^{\pi})^\gamma J(.)$ is a supermartingale for any $\pi \in \mathcal{A}$ with $J(T)=1$. More precisely, if $\bar{J}$ is an $\F$-adapted process such that $(X^{\pi})^\gamma  \bar{J}$ is a supermartingale for any $\pi \in \mathcal{A}$ with $ \bar{J}_T=1$, then for any $t\in[0,T]$, we have $J(t)\leq  \bar{J}_t$ a.s.
\end{Proposition}
From \cite{krasch99}, there exists an optimal strategy $\hat{\pi}\in \Ac$ such that $J(0)=\E[(X_T^{\hat{\pi}})^\gamma]$. From the dynamic programming principle, we have the following optimality criterion:
\begin{Proposition}\label{Critere optimalite}
The following assertions are equivalent:
\begin{enumerate}[i)]
\item $\hat{\pi}$ is an optimal strategy, that is $\E[(X^{\hat{\pi}}_T)^\gamma]=\sup_{\pi\in \Ac}\E[(X^\pi_T)^\gamma]$.
\item The process $(X^{\hat{\pi}})^\gamma J(.)$ is a martingale.
\end{enumerate}  
\end{Proposition}

\noindent The proof of these propositions is given in Appendix \ref{Preuve des propositions}.\\

By Proposition \ref{J(t) plus petite surmartingale}, $J(.)$ is a supermartingale. Hence for any $t \in [0,T]$ we have $\E[J(t)]\leq J(0)<\infty$.

\begin{Proposition}
There exists a c\`{a}d-l\`{a}g modification of $J(.)$ which is denoted by $J$. 
\end{Proposition}

\begin{proof}
 By Proposition \ref{Critere optimalite}, we know that $J(t)= \E[(X_T^{\hat{\pi}})^\gamma|\Fc_t]/(X_t^{\hat{\pi}})^\gamma$ a.s. Which implies the desired result.
\end{proof}

This c\`{a}d-l\`{a}g process is characterized by a BSDE. More precisely, we have:
\begin{Theorem}\label{Proposition BSDE J}
There exist $Z\in L^2_{loc}(W)$ and $U\in L^1_{loc}(M)$ such that $(J, Z, U)$ is the minimal solution\footnote{That is for any solution $(\bar{J}, \bar{Z}, \bar{U})$ of the BSDE (\ref{BSDE puissance}) in $L^{1,+}\times L^2_{loc}(W) \times L^1_{loc}(M)$, we have $J_t \leq \bar{J}_t,~\forall\, t\in [0,T]$ a.s.} in $L^{1,+} \times L^2_{loc}(W) \times L^1_{loc}(M)$ of the following BSDE 
\begin{equation}
\left\{\begin{aligned}
-\,dJ_t~=&  \esssup\limits_{\pi\in\mathcal{A}}\Big\{\gamma \pi_t(\mu_tJ_t+\sigma_tZ_t)+\frac{\gamma(\gamma-1)}{2}\pi_t^2\sigma_t^2J_t +\lambda_t((1+\pi_t\beta_t)^\gamma-1)(J_t+U_t)\Big\}dt \\ 
&-Z_t dW_t - U_tdM_t,\\
J_T~=&\;1.
\end{aligned}\right.
\label{BSDE puissance}
\end{equation}
 There exists a unique optimal strategy such that $J(0)=\E[(X_T^{\hat \pi})^\gamma]$. Moreover, $\hat{\pi}\in \Ac$ is optimal if and only if it attains the essential supremum of the generator in (\ref{BSDE puissance}) $dt\otimes d\P-a.e.$
\end{Theorem}

\noindent The proof of this theorem is postponed in Appendix \ref{preuve du theoreme bsde}.\\

There exists another characterization of the value function $J$ as the limit of processes $(J^k)_{k\in\N}$ as $k$ tends to $+\infty$, when $J^k$ is the value function in the case where the strategies are bounded by $k$:

\begin{Theorem}\label{J limite}
For any $t\in[0,T]$, we have
\begin{equation*}
J_t=\lim\limits_{k\rightarrow \infty} \uparrow J^k(t) ~ a.s.
\end{equation*}
\end{Theorem}

\noindent The proof of this theorem is given in Appendix \ref{preuve du theoreme approximation}. 

This allows to approximate the value function $J$ by numerical computation, since the value functions $(J^k)_{k \in \N}$ are the solution of Lipschitz BSDEs and the results of \cite{boueli08} can be applied.

\section{The partial information case}
\setcounter{equation}{0} \setcounter{Assumption}{0}
\setcounter{Theorem}{0} \setcounter{Proposition}{0}
\setcounter{Corollary}{0} \setcounter{Lemma}{0}
\setcounter{Definition}{0} \setcounter{Remark}{0}

We consider a general filtration which modelizes the information given by the prices $(S_t)_{0 \leq t \leq T}$, the default time $\tau$, but also by other factors. These factors can have in particular an influence on the default probability. We consider an agent on this market, which does not observe all the information but only the information given by the prices and the default time. The underlying Brownian motion, the drift process and the compensator process in the  equation for the asset price are not directly observable. \\

Let $(\Omega,\Fc,\P)$ be a probability triplet and $\F=\{\Fc_t,0\leq t \leq T\}$ a filtration in $\Fc$ satisfying the usual conditions (augmented and right continuous). Suppose that this space is equipped with $W$ and $N$ as in Section \ref{modele}. We also assume there are a risk-free asset and a risky asset on the market. As in Section \ref{modele}, we assume that the price process $S$ evolves according to the following model  
\begin{equation}\label{partiel actif S}
dS_t = S_{t^-} (\mu_tdt + \sigma_t dW_t + \beta_tdN_t),~~ 0\leq t \leq T,
\end{equation}
moreover we assume that
\begin{equation*}
\sigma_t = 
\begin{cases}
\sigma_1(t,S_{t}) & \mbox{if } t \leq \tau \,,  \\
\sigma_2(t,S_{t}, \tau, \beta_\tau)  & \mbox{if } t > \tau \,, 
\end{cases} 
\end{equation*}
and
\begin{equation*}
\beta_t = \beta(t,S_{t^-}) ~\text{if } t \leq \tau \,,  
\end{equation*}
Note that the intensity $\lambda$ and the drift $\mu$ are not necessary observable.
The known functions $\sigma_1(t,x)$ and $\beta(t,x)$ are measurable mappings from $[0,T]\times \R$ into $\R$, and the function $\sigma_2(t,x,s,b)$ is a measurable mapping from $[0,T]\times \R \times \R_+ \times ]-1,\infty [$ into $\R$. We make the hypotheses of Assumption \ref{hypothese coeff} and we add the following assumption:
\begin{Assumption}\emph{\label{partiel hypothese coeff}
The functions $x\sigma_1(t,x)$, $x\sigma_2(t,x,s,b)$ and $x\beta_1(t,x)$ are Lipschitz in $x\in  \R$ , uniformly in $t\in [0,T]$, $s\in \R_+$ and $b\in ]-1,\infty[$.
}\end{Assumption}

We now consider an agent in this market who can observe neither the Brownian motion $W$ nor the drift $\mu$ and the process $\lambda$, but only the asset price process $S$ and the default time $\tau$. We shall denote by $\G=\{\Gc_t,0\leq t \leq T\}$ the $\P$-filtration augmented by the price process $S$ and the default process $N$. The trading strategies are defined as in Section \ref{modele}, but we add the condition that they are $\G$-predictable. We now want to solve the problem of maximization of expected utility from terminal wealth. It is not possible to use directly the  results of the full information case because we do not know the Brownian motion, the drift and the compensator. As in \cite{phaque01}, we begin by an operation of filtering.

\subsection{Filtering} \label{filtrage}  

Recall that we have assumed that $\theta_t=\mu_t / \sigma_t$ is uniformly bounded, therefore the following integrability condition holds
\begin{equation*}
\int_0^T|\theta_t|^2dt<\infty ~ a.s.
\end{equation*}
\noindent Consider the positive martingale defined by $L_0=1$ and $dL_t=-L_t\,\theta_tdW_t$. It is explicitly given by
\begin{equation}
L_t=\exp\Big(-\int_0^t\theta_s dW_s-\frac{1}{2}\int_0^t |\theta_s|^2 ds\Big) \,.
\label{L}
\end{equation}
One can define a probability measure equivalent to $\P$ on $(\Omega,\Fc)$ characterized by
\begin{equation}
\frac{d\,\Q}{d\,\P}\Big|_{\Fc_t}=L_t,~0\leq t \leq T \,.
\label{Q}
\end{equation}
By Girsanov's theorem, the process defined by
\begin{equation}
\tilde{W}_t=W_t+\int_0^t \theta_sds
\label{chang brownien tilde}
\end{equation}
is a $(\Q,\mathbb{F})$-Brownian motion and the compensated martingale $M$ is still a $(\Q,\mathbb{F})$-martingale. The dynamic of $S$ under $\Q$ is given by
\begin{equation}
dS_t= S_{t^-} (\sigma_td\tilde{W}_t + \beta_tdN_t) \,.
\label{eq RN}
\end{equation}
We begin by proving a lemma which will be of paramount  importance in the sequel:
\begin{Proposition}
Under Assumption \ref{hypothese coeff}, the filtration $\mathbb{G}$ is the augmented filtration of $(\tilde{W},N)$.
\label{engendre}
\end{Proposition}

\begin{proof}
Let $\mathbb{F}^{\tilde{W},N}$ be the augmented filtration of $(\tilde{W},N)$. From (\ref{eq RN}), we have
\begin{equation*}
\tilde{W}_t=\int_0^t \sigma_s^{-1} S^{-1}_{s^-} dS_s - \int_0^t \sigma_s^{-1} \beta_s dN_s \,, 
\end{equation*}
for any $t\in [0,T]$, which implies that $\tilde{W}$ is $\mathbb{G}$-adapted and $\mathbb{F}^{\tilde{W},N} \subset \mathbb{G}$. Conversely, under the assumptions on the coefficients, by a classical result of stochastic differential equation (see \cite{pro90}, Theorem V 3.7), the unique solution of (\ref{eq RN}), on $0 \leq t < \tau$, is $\mathbb{F}^{\tilde{W}}$-adapted, and by using a Picard sequence and an iteration we prove that the unique solution of (\ref{eq RN}) is   $\mathbb{F}^{\tilde{W}, N}$-adapted on $\tau \leq t$. Hence $\G \subset \mathbb{F}^{\tilde{W},N}$ and finally $\G = \mathbb{F}^{\tilde{W},N}$.
\end{proof}

Since the processes $\theta$ and $\lambda$ are not $\G$-predictable, it is natural to introduce the $\G$-conditional law of these random variables, say
\begin{equation*}
\tilde{\lambda}_t=\E\big[\lambda_t\big|\mathcal{G}_t\big]~\text{and  }
\tilde{\theta}_t=\E\big[\theta_t\big|\mathcal{G}_t\big] \,.
\end{equation*}
Consider the couple of processes $(\bar{W}, \bar{M})$ defined by
\begin{equation}\label{MbarWbar}
\left\{\begin{aligned}
\bar{W}_t &=\tilde{W}_t-\int_0^t\tilde{\theta}_sds \,,\\
\bar{M}_t&=N_t-\int_0^{t}\tilde{\lambda}_sds \,.
\end{aligned}
\right.
\end{equation}
These are the so-called innovation processes of filtering theory. By classical results in filtering theory (see for example \cite{par89}, Proposition 2.27), we have:
\begin{Proposition}
The process $\bar{M}$ is a $(\Q,\mathbb{G})$-martingale.
\label{M barre}
\end{Proposition}

\begin{proof}
Since the process $N$ and the intensity $\tilde{\lambda}$ are $\mathbb{G}$-adapted, the process $\bar{M}$ is $\mathbb{G}$-adapted. We can write from (\ref{M})
\begin{equation*}
\bar{M}_t=M_t+\int_0^t(\lambda_s-\tilde{\lambda}_s)ds \,.
\end{equation*}
By the law of iterated conditional expectation, it is easy to check that $\bar{M}$ is a $(\Q,\mathbb{G})$-martingale.
 \end{proof}

\begin{Remark}
\emph{From Proposition \ref{engendre} and (\ref{MbarWbar}), the filtration $\mathbb{G}$ is equal to the augmented filtration of $(\tilde{W},\bar{M})$, since $[\bar{M}]_t=N_t$.
\label{WM}}
\end{Remark}

We have also the following property about the process $\Wb$:
\begin{Proposition}
The process
$\bar{W}$ is a $(\P,\mathbb{G})$-Brownian motion.
\end{Proposition}

\begin{proof}
We can write with (\ref{chang brownien tilde})
\begin{equation}
\bar{W}_t=W_t+\int_0^t\sigma_s^{-1}(\mu_s-\tilde{\mu}_s)ds \,,
\label{W bar}
\end{equation}
where $\tilde{\mu}_t=\E\big[\mu_t\big|\mathcal{G}_t\big]$. By Proposition \ref{engendre} and (\ref{MbarWbar}), $\bar{W}$ is $\G$-adapted. Moreover, we have $[\bar{W},\bar{W}]_t=t$ for any $t\in [0,T]$. By the law of iterated conditional expectation, it is easy to check that $\bar{W}$ is a $\G$-martingale. We then conclude by Levy's characterization theorem on Brownian motion (see, e.g., Theorem 3.3.16 in \cite{karshr91}).
 \end{proof}

Denote by $\Lambda$ the $(\Q,\F)$-martingale given by $\Lambda_t=1/L_t$. We then have
\begin{equation*}
\frac{d\,\P}{d\,\Q}\Big|_{\Fc_t}=\Lambda_t,~0\leq t \leq T \,.
\end{equation*}
Let $\tilde{\Lambda}$ be the $(\Q,\G)$-martingale given by $\tilde{\Lambda}_t=\E_\Q [\Lambda_t\big|\mathcal{G}_t ]$.
Recall the classical proposition (see for example \cite{lak98} or \cite{phaque01}), which gives the expression of $\tilde{\Lambda}$:

\begin{Lemma}\label{equation lambda tilde}
We have 
\begin{equation}\label{lambda prop}
\tilde{\Lambda}_t=\exp\Big(\int_0^t\tilde{\theta}_sd\tilde{W}_s-\frac{1}{2}\int_0^t|\tilde{\theta}_s|^2ds\Big) \,.
\end{equation}
\end{Lemma}

\begin{Proposition}
The process $\bar{M}$ is a $(\P,\G)$-martingale.
\end{Proposition}
\begin{proof}
Since $\frac{d\,\P}{d\,\Q}\big|_{\Gc_t}=\tilde{\Lambda}_t$, we can apply Girsanov's theorem and we get that the process $\bar{M}$ is a $(\P,\G)$-martingale.
\end{proof}

By means of innovation processes, we can describe from (\ref{partiel actif S}) and (\ref{W bar}) the dynamics of the partially observed default model within a framework of full observation model
\begin{equation}
\left\{\begin{aligned}
dS_t & =S_{t^-}(\tilde{\mu}_t dt + \sigma_t d\bar{W}_t + \beta_t dN_t) \,,\\
d\bar{M}_t & = dN_t - \tilde{\lambda}_t dt \,,
\end{aligned}\right.
\end{equation}
where $\tilde{\mu}$ and $\tilde{\lambda}$ are $\G$-predictable processes.  \\
Hence, the operations of filtering and control can be put in sequence and thus separated.

\subsection{Optimization problem for the classical utilities}
To apply the results of Section \ref{section puissance}, it is sufficient to have a martingale representation theorem for $(\P,\G)$-martingales with respect to $\Wb$ and $\Mb$. Notice it cannot be directly derived from the usual martingale representation theorem since $\G$ is not equal to the filtration generated by $\Wb$ and $\Mb$.
\begin{Lemma}\label{representation}
Any $(\P,\G)$-local martingale has the representation 
\begin{equation}\label{equation representation2}
m_t=m_0+\int_0^t a_s d\Wb_s + \int_0^t b_s d\Mb_s, ~\forall\,t\in[0,T]~a.s. \,,
\end{equation}
where $a\in L^2_{loc}(\Wb)$ and $b\in L^1_{loc}(\Mb)$. If $m$ is a square integrable martingale, each term on the right-hand side of the representation (\ref{equation representation2}) is square integrable. 
\end{Lemma}
\noindent The proof of this lemma is postponed in Appendix \ref{Preuve du theoreme de representation}.\\

It is now possible to apply the previous results because the price process evolves according to the equation
 \begin{equation*}
\left\{\begin{aligned}
dS_t & = S_{t^-}(\tilde{\mu}_tdt + \sigma_t d\bar{W}_t + \beta_tdN_t),\\
d\bar{M}_t & =dN_t-\tilde{\lambda}_tdt,
\end{aligned}\right.
\end{equation*}
where $\tilde \mu$, $\tilde{\lambda}$ and $\sigma$ are $\G$-adapted, and $\beta$ is $\G$-predictable,
 and there exists a martingale representation theorem for $(\P,\G)$-martingales. We get the following characterization for the value functions and the optimal strategies when they exist.

For the logarithmic utility function, we assume that $\beta^{-1}$ is uniformly bounded, and we have:
\begin{Theorem}
The solution of the optimization problem for the logarithmic utility function is given by 
 \begin{equation*}
 V(x)=\log(x)+\E\Big [\int_0^T\big(\hat{\pi}_t\tilde{\mu}_t-\frac{|\hat{\pi}_t\sigma_t|^2}{2}+\tilde{\lambda}_t\log(1+\hat{\pi}_t\beta_t)\big)dt\Big],
  \end{equation*}
 with $\hat{\pi}$ the optimal trading strategy defined by
  \begin{equation*}
 \hat{\pi}_t=\left\{\begin{aligned}
& \frac{\tilde{\mu}_t}{2\sigma^2_t}-\frac{1}{2\beta_t}+\frac{\sqrt{(\tilde{\mu}_t\beta_t+\sigma^2_t)^2+4\tilde{\lambda}_t\beta^2_t\sigma^2_t}}{2\beta_t\sigma^2_t} ~~\text{ if} ~t \leq \tau~\text{and } \beta_t\neq 0,  \\ 
&\frac{\tilde{\mu}_t}{\sigma^2_t} ~~\text{ if} ~t \leq \tau~\text{and } \beta_t= 0 ~\text{or } t > \tau.
\end{aligned}\right.
 \end{equation*}
\end{Theorem}
Therefore, we can see that the optimal portfolio in the case of partial information can be formally derived from the full information case by replacing the unobservable coefficients $\mu_t$ and $\lambda_t$ by theirs estimates $\tilde{\mu}_t$ and $\tilde{\lambda}_t$. \\

For the power utility function, we have:
\begin{Theorem}
\begin{itemize}
\item Let $(\bar{Y}, \bar{Z}, \bar{U})$ be the minimal solution in $L^{1,+} \times L^2_{loc}(\Wb) \times L^1_{loc}(\Mb)$ of the BSDE (\ref{BSDE puissance})  and $(W,\,M,\,\mu,\,\lambda)$ replaced by $(\Wb,\,\Mb,\,\tilde{\mu},\,\tilde{\lambda})$, then 
\begin{equation*}
\bar{Y}_t=\esssup_{\pi\in\Ac_t}\E\big[(X^{t,\pi}_T)^\gamma\big|\Gc_t\big]~a.s.
\end{equation*}
\item If a strategy $\hat{\pi}\in \Ac$ is optimal for $J_0=\sup_{\pi\in\Ac}\E[(X_T^{\pi})^\gamma]$ then $\hat{\pi}$ attains the essential supremum in the generator of the BSDE (\ref{BSDE puissance}) $dt\otimes d\P$ $a.s.$
\item Moreover the process $\bar{Y}$ is the non-decreasing limit of the process $(\bar{Y}^k)_{k\in\N}$, where $(\bar{Y}^k, \bar{Z}^k, \bar{U}^k)$ is the solution in $\Sc^2\times L^2(\Wb)\times L^2(\Mb)$ of the BSDE (\ref{EDSR lipschitz}) and $(W, M, \mu, \lambda)$ replaced by $(\Wb, \Mb, \tilde{\mu}, \tilde{\lambda})$.
\end{itemize}
\end{Theorem}

\subsection{Optimization problem for the exponential utility function and indifference pricing}
We can also apply the results of \cite{limque09} for the exponential utility function. In this case, we assume that the agent faces some liability, which is modeled by a random variable $\xi$ (for example, $\xi$ may be a contingent claim written on some default events affecting the prices of the underlying assets). We suppose that $\xi$ is a non-negative $\Gc_T$-adapted process (note that all the results still hold under the assumption that $\xi$ is only lower bounded). Without loss of generality we can use a somewhat different notion of trading strategy: $\phi_t$ corresponds to the amount of money invested in the asset at time $t$. The number of shares is $\phi_t / S_t$. With this notation, under the assumption that the trading strategy is self-financing, the wealth process $X^{x,\phi}$ associated with a trading strategy $\phi$ and an initial capital $x$ is equal to 
\begin{equation*}
X^{x,\phi}_t=x+\int_0^t\phi_s\tilde{\mu}_sds+\int_0^t\phi_s\sigma_sdW_s+\int_0^t\phi_s\beta_sdN_s.
\end{equation*}
Our goal is to solve the optimization problem for an agent who buys a contingent claim $\xi$
\begin{equation}
V(x,\xi)=\sup\limits_{\phi\in \mathcal{A}(x)}\E\big[-\exp\big(-\gamma \big( X_T^{x,\phi}+\xi\big)\big)\big]=\exp(-\gamma x)V(0,\xi),
\label{pb contingent}
\end{equation}
where $\mathcal{A}(x)$ is defined by:

\begin{Definition}\emph{ \label{portefeuille admissible}
The set of admissible trading strategies $\Ac(x)$ consists of all $\G$-predictable processes $\phi=(\phi_t)_{0\leq t\leq T}$, which satisfy $\int_0^T|\phi_t\sigma_t|^2ds+\int_0^T|\phi_t\beta_t|^2\tilde{\lambda}_tdt<\infty,~\P-a.s.$ and such that for any $\phi$ fixed and any $t\in[0,T]$, there exists a constant $K_{t,\pi}$ such that for any $s \in [t,T]$, we have $X^{t,\pi}_s \geq K_{t,\pi},~\P-a.s.$
}\end{Definition}

To solve this problem, it is sufficient to study the case $x=0$. For that we give a dynamic extension of the initial problem as in Section \ref{section puissance}. For any initial time $t\in [0,T]$, we define the value function $J^\xi(t)$ by the following random variable
\begin{equation*}
J^\xi(t)= \essinf \limits_{\phi \in \mathcal{A}_t} \E\big[\exp\big(-\gamma \big(X_T^{t,0,\phi}+\xi\big)\big)\big| \mathcal{G}_t\big],
\end{equation*}
with $\mathcal{A}_t$ is the admissible portfolio strategies set defined by:

\begin{Definition}\emph{
The set of admissible trading strategies $\Ac_t$ consists of all $\G$-predictable processes $\phi=(\phi_s)_{t\leq s\leq T}$, which satisfy $\int_t^T|\phi_s\sigma_s|^2ds+\int_t^T|\phi_s\beta_s|^2\tilde{\lambda}_sds<\infty,~\P-a.s.$ and such that for any $\phi$ fixed and any $s\in[t,T]$, there exists a constant $K_{s,\pi}$ such that for any $u \in [s,T]$, we have $X^{s,\pi}_u \geq K_{s,\pi},~\P-a.s.$
}\end{Definition}
\noindent We introduce the two following sets:
\begin{itemize}
\item $\Sc^{+,\infty}$ is the set of positive $\G$-adapted $\P$-essentially bounded c\`ad-l\`ag processes on $[0,T]$.
\item $\Ac^2$ is the set of the increasing adapted c\`{a}d-l\`{a}g processes $K$ such that $K_0=0$ and $\E|K_T|^2<\infty$.
\end{itemize}
By applying the results of the companion paper Lim and Quenez \cite{limque09}, we get the following characterizations of the value function:
\begin{Theorem}\label{theoreme utilite exponentielle}
\begin{itemize}
\item Let $(\bar{Y},\bar{Z},\bar{U},\bar{K})$ be the maximal solution\footnote{That is for any solution $(\bar{J},\bar{Z},\bar{U},\bar{K})$ of the BSDE (\ref{EDSR exponentielle}) in $\Sc^{+,\infty} \times L^2(\bar{W}) \times L^2(\bar{M})\times \Ac^2$, we have $ \bar{J}_t\leq J_t,~\forall\, t\in [0,T],~\P-a.s.$} in $\Sc^{+,\infty} \times L^2(\bar{W}) \times L^2(\bar{M})\times \Ac^2$ of
\begin{equation}\label{EDSR exponentielle}
\left\{
\begin{aligned}
-\,d\bar{Y}_t=&-\bar{Z}_td\Wb_t-\bar{U}_td\Mb_t-d\bar{K}_t+\essinf\limits_{\phi\in \mathcal{A}}\Big\{ \frac{\gamma^2}{2}|\phi_t\sigma_t|^2\bar{Y}_t   -\gamma\phi_t(\tilde{\mu}_t\bar{Y}_t+\sigma_t\bar{Z}_t)\\
& -\big(1-e^{-\gamma \phi_t\beta_t}\big)(\tilde{\lambda}_t\bar{Y}_t+\tilde{\lambda}_t\bar{U}_t)\Big\}dt,\\
 \bar{Y}_T=&\;\exp(-\gamma \xi),
\end{aligned}
\right.
\end{equation}
then $\bar{Y}_t=\bar{J}^\xi(t),~\P-a.s.$, for any $t\in [0,T]$.
\item $\bar{J}^\xi(t) = \lim_{n \rightarrow \infty} \downarrow \bar{J}^{ \xi, k}(t)$, with $\bar{J}^{\xi, k}(t) = \essinf_{\phi \in \Ac^k_t} \E[\exp(-\gamma (X_T^{t,0,\phi}+\xi))| \mathcal{G}_t]$ and $\Ac^k_t$ is the set of strategies of $\Ac_t$ uniformly bounded by $k$.
\item Let $(\bar{Y}^k,\bar{Z}^k,\bar{U}^k)$ be the unique solution in $\Sc^2\times L^2(\Wb)\times L^2(\Mb)$ of the following BSDE
\begin{equation}\label{edsr k}
\left\{
\begin{aligned}
-\,d \bar{Y}^{k}_t=&-\bar{Z}^{k}_td\bar{W}_t-\bar{U}^{k}_td\bar{M}_t+\essinf\limits_{\phi\in \mathcal{A}^k}\Big\{ \frac{\gamma^2}{2}|\phi_t\sigma_t|^2\bar{Y}^{k}_t  -\gamma\phi_t(\tilde{\mu}_t \bar{Y}^{k}_t+\sigma_t\bar{Z}^k_t)\\
&-(1-e^{-\gamma\phi_t\beta_t})( \tilde{\lambda}_t\bar{Y}^{k}_t+\tilde{\lambda}_t\bar{U}^{k}_t)\Big\}dt , \\
 \bar{Y}^{k}_T=&\;\exp(-\gamma \xi),
\end{aligned}
\right.
\end{equation}
then $\bar{Y}^k_t=\bar{J}^{\xi,k}(t),~\P-a.s.$, for any $t\in [0,T]$.
\end{itemize}
\end{Theorem}

We can now define the indifference pricing of the contingent claim $\xi$. The Hodges approach to pricing of unhedgeable claims is a utility-based approach and can be summarized as follows: the issue at hand is to assess the value of some (defaultable) claim $\xi$ as seen from the perspective of an investor who optimizes his behavior relative to some utility function, in our case we use the exponential utility function. The investor     
has two choices: 
\begin{itemize}
\item he only  invests in the risk-free asset and in the risky assets, in this case the associated optimization problem is
\begin{equation*}
V(x,0)=\sup\limits_{\phi\in\Ac(x)}\E\big[-\exp\big(-\gamma\big(X_T^{x,\phi}\big)\big)\big],
\end{equation*}  
\item he also invests in the contingent claim, whose price is $\bar{p}$ at $0$, in this case the associated optimization problem is
\begin{equation*}
V(x-\bar{p},\xi)=\sup\limits_{\phi\in\Ac(x-\bar p)}\E\big[-\exp\big(-\gamma\big(X_T^{x-\bar{p},\phi}+\xi\big)\big)\big].
\end{equation*}
\end{itemize}  
\begin{Definition}\emph{
For a given initial capital $x$, the Hodges buying price of a defaultable claim $\xi$ is the price $\bar{p}$ such that the investor's value functions are indifferent between holding and not holding the contingent claim, i.e.
\begin{equation*}
V(x,0)=V(x-\bar{p},\xi).
\end{equation*}
}\end{Definition}
 
The Hodges price $\bar{p}$ can be derived explicitly by applying the results of Theorem \ref{theoreme utilite exponentielle}. If the agent buys the contingent claim at the price $\bar{p}$ and invests the rest of his wealth in the risk-free asset and in the risky assets, the value function is equal to
\begin{equation*}
V(x-\bar{p},\xi)=-\exp(-\gamma(x-\bar{p}))\bar{J}^\xi(0).
\end{equation*}
If he invests all his wealth in the risk-free asset and in the risky assets, the value function is equal to
\begin{equation*}
V(x,0)=-\exp(-\gamma x)\bar{J}^0(0).
\end{equation*}
 The Hodges price for the contingent claim $\xi$ is clearly given by the formula
 \begin{equation*}
\bar{p}=\frac{1}{\gamma}\ln\Big(\frac{\bar{J}^0(0)}{\bar{J}^\xi(0)}\Big).
 \end{equation*}

\begin{Remark}
\emph{If we restrict the admissible strategies to the bounded set $\Ac^k$, the indifference price $\bar{p}^k$ can also be defined by the same method. More precisely,
 \begin{equation*}
 \bar{p}^k=\frac{1}{\gamma}\ln\Big(\frac{\bar{J}^{0,k}(0)}{\bar{J}^{\xi,k}(0)}\Big),
 \end{equation*}
 where $\bar{J}^{\xi,k}(0)$ is defined in Theorem \ref{theoreme utilite exponentielle}.\\
 Note that
 \begin{equation*}
 \bar{p}=\lim\limits_{k\rightarrow \infty}\bar{p}^k.
 \end{equation*}
 This allows to approximate the indifference price by numerical computation, since the value functions $(\bar{J}^{\xi,k}(t))_{k\in \N}$ are the solution of a Lipschitz BSDE and the results of \cite{boueli08} can be applied.}
\end{Remark}

We assume that there are two kinds of agents in the market: the insider agents and the classical agents. We define the information price $d$ for a contingent claim as the difference between the buying price for a classical agent and the buying price for an insider agent. The buying price, if the agent knows the full information, is defined by (see \cite{limque09})
\begin{equation*}
p=\frac{1}{\gamma}\ln\Big(\frac{J^0(0)}{J^\xi(0)}\Big),
\end{equation*}
where $(J^\xi,Z,U,K)$ is the maximal solution of the BSDE (\ref{EDSR exponentielle}) with $(\Wb,\,\Mb,\,\tilde{\mu},\,\tilde{\lambda})$ replaced by $(W,\,M,\,\mu,\,\lambda)$.\\
Then the benefit of an insider agent who has a full information is the information price
\begin{equation*}
d=\bar{p}-p.
\end{equation*}
This price can be computed as the limit of the information prices $(d^k)_{k\in \N}$, where $d^k$ is the information price if we restrict the admissible strategies to the bounded set $\Ac^k$
\begin{equation*}
d^k=\frac{1}{\gamma}\Big(\ln\Big(\frac{\bar{J}^{0,k}(0)}{J^{0,k}(0)}\Big)-\ln\Big(\frac{\bar{J}^{\xi,k}(0)}{J^{\xi,k}(0)}\Big)\Big),
\end{equation*}
where $J^{\xi,k}$ is the solution of the BSDE (\ref{edsr k}) with $(\Wb,\,\Mb,\,\tilde{\mu},\,\tilde{\lambda})$ replaced by $(W,\,M,\,\mu,\,\lambda)$.\\
Then we have
\begin{equation*}
d=\lim\limits_{k\rightarrow \infty} d^k.
\end{equation*}

\vspace{1cm}

\appendix
\begin{center}{\Large{\textbf{Appendix}}}\end{center}

\section{Proof of Propositions \ref{J(t) plus petite surmartingale} and \ref{Critere optimalite}} \label{Preuve des propositions}
\setcounter{equation}{0} \setcounter{Assumption}{0}
\setcounter{Theorem}{0} \setcounter{Proposition}{0}
\setcounter{Corollary}{0} \setcounter{Lemma}{0}
\setcounter{Definition}{0} \setcounter{Remark}{0}
The proofs of these two propositions are based on the following lemma:
\begin{Lemma}\label{stable}
The set $\left\{J^\pi_t,~\pi\in\Ac_t\right\}$ is stable by supremum for any $t\in [0,T]$, i.e. for any $\pi^1,~\pi^2  \in  \Ac_t$, there exists $ \pi \in \Ac_t$ such that $ J^{\pi}_t =J^{\pi^1}_t \vee J^{\pi^2}_t $.\\
Furthermore, there exists a sequence $(\pi^n)_{n \in \mathbb{N}} \in \Ac_t$ for any $t\in [0,T]$, such that
\begin{equation*}
J(t)=\lim\limits_{n \rightarrow \infty} \uparrow J^{\pi^n}_t ,~\P-a.s.
\end{equation*}
\end{Lemma}

\begin{proof}
Let us introduce the set $E=\{J ^{\pi^1}_t  \geq J ^{\pi^2}_t)\}$ which belongs to $\Fc_t$. Let us define the strategy $\pi$ by the formula $\pi_s=\pi_s^1\mathds{1}_{E}+\pi_s^2\mathds{1}_{E^c}$ for any $s\in[t,T]$. It is obvious that $\pi\in \Ac_t$. And by construction of $\pi$, it is clear that $J ^{\pi}_t =J ^{\pi^1}_t \vee J ^{\pi^2}_t $.\\
The second part of the lemma follows by classical results on the essential supremum (see \cite{nev75}).
\end{proof}

We first prove that the process $(X^{\pi})^\gamma J(.)$ is a supermartingale for any $\pi \in \mathcal{A}$. For that it is sufficient to show for any $s\leq t$ that
\begin{equation*}
\E\big[(X_t^{s,\pi})^\gamma J(t) \big|\Fc_s\big] \leq J(s)~a.s.
\end{equation*}
By Lemma \ref{stable}, there exists a sequence $(\pi^n)_{n\in\N}$ of $\Ac_t$ such that $J(t)=\lim\uparrow J^{\pi^n}_t$ a.s. We define the strategy $\tilde{\pi}^n$ by $\tilde{\pi}^n_u=
\pi_u \1_{[s,t]}(u)+\pi^n_u \1_{]t,T]}(u)$, which is clearly admissible. By the monotone convergence theorem and using the definition of $J(.)$, one can easily show that
\begin{equation*}
\E\big[(X_t^{s,\pi})^\gamma J(t)\big|\Fc_s\big]=\lim_{n\rightarrow \infty}\uparrow \E\big[(X_T^{s,\tilde{\pi}^n})^\gamma\big|\Fc_s\big]\leq J(s)~a.s.
\end{equation*}
Hence, the process $(X^{\pi})^\gamma J(.)$ is a supermartingale for any $\pi\in\Ac$. \\

Second, we prove that $J(.)$ is the smallest process satisfying $(X^{\pi})^\gamma J(.)$ is a supermartingale for any $\pi \in \mathcal{A}$. For that we suppose that $\bar{J}$ is an $\F$-adapted process such that $(X^{\pi})^\gamma \bar{J}$ is a supermartingale for any $\pi \in \mathcal{A}$ with $ \bar{J}_T=1$. Fix $t\in[0,T]$. For any $\pi\in \mathcal{A}$, we have $\E[(X_T^{\pi})^\gamma|\mathcal{F}_t] \leq (X_t^{\pi})^\gamma \bar{J}_t$ a.s. This inequality is equivalent to $\E[(X_T^{t,\pi})^\gamma|\Fc_t]\leq \bar{J}_t$. Which implies 
\begin{equation*}
\esssup\limits_{\pi \in \Ac_t}\E\big[(X_T^{t,\pi})^\gamma\big|\mathcal{F}_t\big] \leq \bar{J}_t~a.s.\,,
\end{equation*}
which clearly gives that $J_t\leq  \bar{J}_t$ a.s.\\

At last, we prove the optimality criterion, that is Proposition \ref{Critere optimalite}. Suppose that the strategy $\hat{\pi}$ is an optimal strategy, hence we have 
\begin{equation*}
J(0)=\sup\limits_{\pi\in\mathcal{A}}\E\big[\big(X_T^{\pi}\big)^\gamma\big]=\E\big[\big(X_T^{\hat{\pi}}\big)^\gamma\big].
\end{equation*}
As the process $(X^{\hat{\pi}})^\gamma J(.)$ is a supermartingale by Proposition \ref{J(t) plus petite surmartingale} and that $J(0)=\E[(X_T^{\hat{\pi}})^\gamma ]$, the process $(X^{\hat{\pi}})^\gamma J(.)$ is a martingale.\\
To show the converse, suppose that the process $(X^{\hat{\pi}})^\gamma J(.)$ is a martingale, then $\E[(X_T^{\hat{\pi}})^\gamma ]=J(0)$. Moreover $\E[(X_t^\pi)^\gamma J(t)]\leq J(0)$ for any $\pi\in\Ac$ by Proposition \ref{J(t) plus petite surmartingale}. Which implies that
\begin{equation*}
J(0)=\sup\limits_{\pi\in\mathcal{A}}\E\big[\big(X_T^{\pi}\big)^\gamma\big]=\E\big[\big(X_T^{\hat{\pi}}\big)^\gamma\big].
\end{equation*}

\section{Proof of Theorem \ref{Proposition BSDE J}}\label{preuve du theoreme bsde}
\setcounter{equation}{0} \setcounter{Assumption}{0}
\setcounter{Theorem}{0} \setcounter{Proposition}{0}
\setcounter{Corollary}{0} \setcounter{Lemma}{0}
\setcounter{Definition}{0} \setcounter{Remark}{0}

The proof of this theorem is based on Propositions \ref{J(t) plus petite surmartingale} and \ref{Critere optimalite}, on Doob-Meyer's decomposition and on the martingale representation theorem.\\

Since the process $J$ is a supermartingale, it can be written under the following form by using Doob-Meyer's decomposition (see \cite{delmey80}) and the martingale representation theorem
\begin{equation}\label{Doob-Meyer puissance}
dJ_t = Z_t dW_t + U_t dM_t  - dA_t \,,
\end{equation}
with $Z\in L^2_{loc}(W)$, $U\in L^1_{loc}(M)$ and $A$ is a non-decreasing $\F$-adapted process with $A_0=0$.\\ 
From product rule, the derivative of process $(X_t^{\pi})^\gamma J$ can be written under the form 
\begin{equation*}
d((X_t^{\pi})^\gamma J_t)=(X_{t^-}^{\pi})^\gamma \big(dA^\pi_t+ dM^\pi_t\big) \,,
\end{equation*}
with $A_0^\pi=0$ and
\begin{equation}\label{expression A et M}
\left\{\begin{aligned}
dA_t^\pi & =\Big[\gamma \pi_t(\mu_tJ_t+\sigma_t Z_t)+\frac{\gamma(\gamma-1)}{2}\pi_t^2\sigma_t^2J_t+\lambda_t((1+\pi_t\beta_t)^\gamma-1)(J_t+U_t)\Big]dt-dA_t \,,\\
dM_t^\pi & =(\gamma \pi_t\sigma_tJ_t+Z_t)dW_t+(U_t+((1+\pi_t\beta_t)^\gamma-1)(J_t+U_t))dM_t \,.
\end{aligned}\right.
\end{equation}
From Proposition \ref{J(t) plus petite surmartingale}, we have $dA_t^\pi\leq 0$ for any $\pi\in \Ac$, which implies
\begin{equation*}
dA_t\geq \esssup\limits_{\pi\in \Ac}\Big\{\gamma \pi_t(\mu_tJ_t+ \sigma_tZ_t )+\frac{\gamma(\gamma-1)}{2}\pi_t^2\sigma_t^2J_t+\lambda_t((1+\pi_t\beta_t)^\gamma-1)(J_t+U_t)\Big\}dt \,.
\end{equation*}
From \cite{krasch99}, there exists an optimal strategy $\hat{\pi} \in \Ac$ to the optimization problem and from Proposition \ref{Critere optimalite}, we get 
\begin{equation*}
dA_t=\Big[\gamma \hat{\pi}_t(\mu_tJ_t+ \sigma_tZ_t )+\frac{\gamma(\gamma-1)}{2}\hat{\pi}_t^2\sigma_t^2J_t+\lambda_t((1+\hat{\pi}_t\beta_t)^\gamma-1)(J_t+U_t)\Big]dt \,.
\end{equation*}
Which imply that
\begin{equation} \label{annexe egalite A}
dA_t= \esssup\limits_{\pi\in \Ac}\Big\{\gamma \pi_t(\mu_tJ_t+ \sigma_tZ_t )+\frac{\gamma(\gamma-1)}{2}\pi_t^2\sigma_t^2J_t+\lambda_t((1+\pi_t\beta_t)^\gamma-1)(J_t+U_t)\Big\}dt \,.
\end{equation}
Therefore the process $(J, Z, U)$ is a solution of the BSDE (\ref{BSDE puissance}). \\

We now prove that it is the minimal solution.
Let $(\bar{J}, \bar{Z}, \bar{U})$ be a solution of the BSDE (\ref{BSDE puissance}). Let us prove that $(X^{\pi})^\gamma \bar{J}$ is a supermartingale for any $\pi \in \mathcal{A}$. From the product rule, we can write the derivative of this process under the form
\begin{equation}\label{derive 1}
d\left(\left(X_t^{\pi}\right)^\gamma \bar{J}_t\right)=(X_{t^-}^{\pi})^\gamma\left[d\bar{M}^\pi_t+d\bar{A}^\pi_t-d\bar{A}_t\right] \,,
\end{equation}
where $\bar{A}$ and $\bar{M}^\pi$ are given by (\ref{annexe egalite A}) and \ref{expression A et M} with $(J, Z, U)$ replaced by $(\bar{J}, \bar{Z}, \bar{U})$, and $\bar{A}_0^\pi=0$ and 
\begin{equation*}
d\bar{A}_t^\pi = \Big[ \gamma \pi_t(\mu_t\bar{J}_t+\sigma_t\bar{Z}_t)+\frac{\gamma(\gamma-1)}{2}\pi_t^2\sigma_t^2\bar{J}_t+\lambda_t((1+\pi_t\beta_t)^\gamma-1)(\bar{J}_t+\bar{U}_t)\Big]dt \,.
\end{equation*}
By integrating (\ref{derive 1}), we get
\begin{equation*}
(X ^{\pi}_t)^\gamma \bar{J}_t-\bar{J}_0=\int_0^t(X^{\pi}_{s^-})^\gamma d\bar{M}^\pi_s-\int_0^t(X_{s}^{\pi})^\gamma (d\bar{A}_s-d\bar{A}_s^\pi)\,.
\end{equation*}
As $d\bar{A}_s\geq d\bar{A}_s^\pi$, we have $\int_0^t(X^{\pi}_{s^-})^\gamma d\bar{M}^\pi_s\geq (X_t^{\pi})^\gamma \bar{J}_t-\bar{J}_0\geq -\bar{J}_0$. It implies that $\bar{M}^\pi$ is a supermartingale, since it is a lower bounded local martingale. Hence, the process $(X^{\pi})^\gamma \bar{J}$ is a supermartingale for any $\pi \in \mathcal{A}$, because it is the sum of a supermartingale and a non-increasing process. Proposition \ref{J(t) plus petite surmartingale} implies that $J_t \leq \bar{J}_t,~\forall\,t\in [0,T]$ a.s., which ends this proof.

\section{Proof of Theorem \ref{J limite}}\label{preuve du theoreme approximation}
\setcounter{equation}{0} \setcounter{Assumption}{0}
\setcounter{Theorem}{0} \setcounter{Proposition}{0}
\setcounter{Corollary}{0} \setcounter{Lemma}{0}
\setcounter{Definition}{0} \setcounter{Remark}{0}

We first remark that $J^k$ satisfies the following property:

\begin{Lemma}\label{Jk surmartingale}
The process $J^k$ is the smallest $\F$-adapted process such that $(X^{\pi})^\gamma J^k$ is a supermartingale for any $\pi \in \Ac^k$ with $J^k_t=1$.
\end{Lemma}

\noindent To prove this lemma, we use exactly the same arguments as in the proof of Proposition \ref{J(t) plus petite surmartingale}, since Lemma \ref{stable} is still true with $\Ac^k_t$ instead of $\Ac_t$. \\

Fix $t\in[0,T]$. It is obvious with the definition of sets $\Ac_t$ and $\mathcal{A}^k_t$ that $\mathcal{A}^k_t \subset \Ac_t$ for each $k\in \N$ and hence
\begin{equation}\label{inegalite J Jk}
 J^k_t\leq J_t~a.s.
\end{equation}
Moreover, since $\Ac^k_t \subset \Ac^{k+1}_t$ for each $k\in \N$, it follows that the positive sequence $(J^k)_{k\in \mathbb{N}}$ is non-decreasing. Let us define the random variable
\begin{equation*}
\tilde{J}(t)=\lim\limits_{k\rightarrow \infty} \uparrow J^k_t~a.s.
 \end{equation*}
It is obvious that $ \tilde{J}(t) \leq J_t$ a.s. from (\ref{inegalite J Jk}) and this holds for any $t\in[0,T]$. It remains to prove that for any $t\in[0,T]$, $J_t\leq \tilde{J}(t)$ a.s. As in the proof of Theorem 4.2 of the companion paper \cite{limque09}, we first prove that the process $\tilde{J}(t^+)$ is c\`{a}d-l\`{a}g and satisfies $\tilde{J}(t^+)\leq \tilde{J}(t)$ a.s. The process $((X_t^\pi)^\gamma \tilde{J}(t^+))$ is a supermartingale for any bounded strategy $\pi\in \Ac$. In the sequel, we shall denote $\bar{J}_t$ instead of $\tilde{J}(t^+)$. We now prove that $\bar{J}_t\geq J_t,~\forall\,t\in [0,T]$ a.s. Since $\bar{J}$ is a c\`{a}d-l\`{a}g supermartingale, it admits the following Doob-Meyer's decomposition
\begin{equation*}
d\bar{J}_t=\bar{Z}_tdW_t + \bar{U}_tdM_t - d\bar{A}_t \,,
\end{equation*}
with $\bar{Z}\in L^2_{loc}(W)$, $\bar{U}\in L^1_{loc}(M)$ and $\bar{A}$ is a non-decreasing $\G$-adapted process with $\bar{A}_0=0$. As before, we use the fact that the process $(X^\pi)^\gamma\bar{J}$ is a supermartingale for any bounded strategy $\pi\in \Ac$ to give some conditions satisfied by the process $\bar{A}$. Let $\pi\in \Ac$ be a uniformly bounded strategy, the product rule gives
\begin{equation} \label{derive 2}
d((X_t^{\pi})^\gamma \bar{J}_t)=(X_{t^-}^{\pi})^\gamma\big(d\bar{A}^\pi_t+d\bar{M}^\pi_t\big) \,,
\end{equation}
where $\bar{A}^\pi$ and $\bar{M}^\pi$ are given by (\ref{expression A et M}) with $(J, Z, U, A)$ replaced by $(\bar{J}, \bar{Z}, \bar{U},  \bar{A})$.

Let $\bar{\mathcal{A}}_t$ be the subset of uniformly bounded strategies of $\Ac_t$. Since the process $(X^\pi)^\gamma\bar{J}$ is a supermartingale for any $\pi \in \bar{\Ac}$, we have 
 \begin{equation}\label{inegalite Abar}
d\bar{A}_t\geq\esssup_{\pi\in\bar{\Ac}}\Big\{\gamma \pi_t(\mu_t\bar{J}_t+\sigma_t \bar{Z}_t)+\frac{\gamma(\gamma-1)}{2}\pi_t^2\sigma_t^2\bar{J}_t+\lambda_t((1+\pi_t\beta_t)^\gamma-1)(\bar{J}_t+\bar{U}_t)\Big\}dt \,.
 \end{equation}
 It is not possible to give an exact expression of $\bar{A}_t$ as in the previous proof, because we do not know if $\hat{\pi} \in \bar{\mathcal{A}}$. But this inequality is sufficient for the demonstration.
Now, the following equality holds $dt\otimes d\P$ $a.s.$
\begin{multline}\label{esssupA=esssupAbar}
\esssup\limits_{\pi\in \bar{\mathcal{A}}}\Big\{\gamma \pi_t(\mu_t\bar{J}_t+\sigma_t \bar{Z}_t)+\frac{\gamma(\gamma-1)}{2}\pi_t^2\sigma_t^2\bar{J}_t+\lambda_t((1+\pi_t\beta_t)^\gamma-1)(\bar{J}_t+\bar{U}_t)\Big\}=\\
\esssup\limits_{\pi\in \mathcal{A}}\Big\{\gamma \pi_t(\mu_t\bar{J}_t+\sigma_t \bar{Z}_t)+\frac{\gamma(\gamma-1)}{2}\pi_t^2\sigma_t^2\bar{J}_t+\lambda_t((1+\pi_t\beta_t)^\gamma-1)(\bar{J}_t+\bar{U}_t) \Big\} \,.
\end{multline}
We now want to show that $(X^{\pi})^\gamma \bar{J}$ is a supermartingale for any $\pi\in\Ac$. Fix $\pi\in\Ac$ (not necessarily uniformly bounded), we get
\begin{equation*}
(X_t^{\pi})^\gamma \bar{J}_t-\bar{J}_0=\int_0^t(X_{s^-}^{\pi})^\gamma d\bar{M}_s^\pi+\int_0^t(X_{s}^{\pi})^\gamma d\bar{A}_s^\pi \,,
\end{equation*}
with $\bar{A}^\pi$ and $\bar{M}^\pi$ given by (\ref{expression A et M}) with $(J, Z, U, A)$ replaced by $(\bar{J}, \bar{Z}, \bar{U}, \bar{A})$.

 Inequality (\ref{inegalite Abar}) and equality (\ref{esssupA=esssupAbar}) imply that $d\bar{A}_t^\pi\leq 0$ a.s. Therefore, we have
\begin{equation*}
\int_0^t(X_{s^-}^{\pi})^\gamma d\bar{M}_s^\pi \geq (X_t^{\pi})^\gamma \bar{J}_t -\bar{J}_0\geq  -\bar{J}_0 \,.
\end{equation*}
Thus, $\bar{M}^\pi$ is a supermartingale, since it is a lower bounded local martingale. As $\bar{M}^\pi$ is a supermartingale and $\bar{A}^\pi$ is non-increasing, the process $(X^{\pi})^\gamma \bar{J}$ is a supermartingale and this holds for any $\pi \in \Ac$. Since $J$ is the smallest process (see Proposition \ref{J(t) plus petite surmartingale}) satisfying these properties, we have $J_t\leq \bar{J}_t$ a.s. Which ends the proof.

\section{Proof of Lemma \ref{representation}}\label{Preuve du theoreme de representation}
\setcounter{equation}{0} \setcounter{Assumption}{0}
\setcounter{Theorem}{0} \setcounter{Proposition}{0}
\setcounter{Corollary}{0} \setcounter{Lemma}{0}
\setcounter{Definition}{0} \setcounter{Remark}{0}

\noindent First, recall Bayes formula: for any $t\in [0,T]$ and $X \in L^1(\Omega,\mathcal{F}_t,\P)$, one has
\begin{equation}
\E\big[X\big|\mathcal{G}_t\big]=\frac{\E_\Q\big[\Lambda_tX\big|\mathcal{G}_t\big]}{\tilde{\Lambda}_t}.
\label{Bayes}
\end{equation}
Let $\xi$ be the optional projection of the $\P$-martingale $L$ to $\G$, so
\begin{equation*}
\xi_t=\E\big[L_t\big|\mathcal{G}_t\big].
\end{equation*}
By applying relation (\ref{Bayes}) to $X=L_t$, we immediately obtain $\xi_t=1/\tilde{\Lambda}_t$ and thus 
\begin{equation*}
\xi_t=\exp\Big(-\int_0^t\tilde{\theta}_sd\bar{W}_s-\frac{1}{2}\int_0^t|\tilde{\theta}_s|^2ds\Big).
\end{equation*}
Let $m$ be a $(\P,\mathbb{G})$-local martingale. From Bayes rule, the process $\tilde{m}$ defined by 
\begin{equation*}
\tilde{m}_t=m_t\xi_t^{-1}, ~~0\leq t \leq T,
\end{equation*}
is a $(\Q,\mathbb{G})$-local martingale. From Remark \ref{WM} and Lemma \ref{representation}, there exists a couple of processes $(\tilde{a},\tilde{b})$ with $\tilde{a}\in L^2_{loc}(\Wt)$ and $\tilde{b}\in L^1_{loc}(\Mb)$ such that
\begin{equation*}
\tilde{m}_t=\int_0^t\tilde{a}'_sd\tilde{W}_s+\int_0^t\tilde{b}'_sd\bar{M}_s,  ~~0\leq t \leq T.
\end{equation*}
By It\^{o}'s formula applied to $m_t=\tilde{m}_t\xi_t$, definition of $\Wb$ and $\Mb$ (see (\ref{MbarWbar})), we obtain that
\begin{equation*}
m_t=\int_0^ta'_sd\bar{W}_s+\int_0^tb'_sd\bar{M}_s,
\end{equation*}
\noindent with $a_t=\xi_t\tilde{a}_t-\tilde{m}_{t}\xi_t\tilde{\rho}_t$ and $b_t=\xi_{t^-}\tilde{b}_t$.

\end{document}